\documentclass[letter]{article}
\usepackage{spconf,amsmath,graphicx,cite}
\usepackage{multirow}

\usepackage{amsfonts,amssymb,amsthm,algorithmic}
\usepackage{bm}
\usepackage{color}
\theoremstyle{definition}
\newtheorem{lemma}{Lemma}
\newtheorem*{lemma*}{Lemma}

\newtheorem{prop*}{Proposition*}

\newtheorem{definition*}{Definition*}

\newtheorem{note*}{Note*}
\newtheorem{theorem}[lemma]{Theorem}
\newtheorem{theorem*}{Theorem*}

\newtheorem{col*}{Corollary*}

\title{FastFCA: a joint diagonalization based fast algorithm for audio source separation using a full-rank spatial covariance model
} 
\name{Nobutaka Ito, Shoko Araki, Tomohiro Nakatani}
\address{NTT Communication Science Laboratories, NTT Corporation, Kyoto, Japan\\
\{ito.nobutaka, araki.shoko, nakatani.tomohiro\}@lab.ntt.co.jp
}
\newtheorem{algorithm}{Algorithm}
\begin{document}
\maketitle
\ninept
\begin{abstract}
A source separation method using a full-rank spatial covariance model has been proposed by Duong {\it et al.} [``Under-determined Reverberant Audio Source Separation Using a Full-rank Spatial Covariance Model,'' {\it IEEE Trans. ASLP}, vol.~18, no.~7, pp.~1830--1840, Sep. 2010], which is referred to as {\it full-rank spatial covariance analysis (FCA)} in this paper.
Here we propose a fast algorithm for estimating the model parameters of the FCA, which is named {\it FastFCA}, and applicable to the two-source case.
Though quite effective in source separation, the conventional FCA has a major drawback of expensive computation.
Indeed, the conventional algorithm for estimating the model parameters of the FCA
 requires {\it frame-wise} matrix inversion and matrix multiplication.
Therefore, the conventional FCA
may be infeasible in applications with restricted computational resources.
In contrast, the proposed FastFCA
 bypasses matrix inversion and matrix multiplication owing to joint diagonalization based on the generalized eigenvalue problem.
Furthermore, the FastFCA is strictly equivalent to the conventional algorithm. 
An experiment has shown that the FastFCA was over 250 times faster than the conventional algorithm with virtually the same source separation performance.
\end{abstract}
\section{Introduction}

Many audio source separation methods take a probabilistic approach, in which a probabilistic model of
 observed mixtures is designed
and some model parameters pertinent to the sources are estimated.
In such an approach, source separation performance
is largely dictated by precision of the probabilistic model.
In many conventional models, such as that in the well-known independent component analysis (ICA)~\cite{Cardoso1999,Lee1999,Hyvarinen2001,Sawada2005SE}, the acoustic transfer characteristics of each source signal are modeled by a time-invariant
steering vector. 
In contrast, Duong {\it et al.}~\cite{Duong2010} have proposed modeling the acoustic transfer characteristics of each source signal by a full-rank matrix called a {\it
spatial covariance matrix}.
The latter model can properly take account of reverberation, fluctuation of source locations, deviation from the ideal point-source model, {\it etc.}, whereby
realizing effective source separation in the real world.
We call this method {\it full-rank spatial covariance analysis (FCA)}.

A major limitation of the conventional FCA is expensive computation.
Indeed, the conventional algorithm for estimating the model parameters of the FCA
 computes matrix inverses and matrix products {\it frame-wise}.
Therefore, the conventional FCA
may be infeasible in applications with restricted computational resources.
Such applications may include hearing aids, distributed microphone arrays, and online speech enhancement.

To overcome this limitation, here  we propose a fast algorithm for estimating the model parameters of the FCA, which is named {\it FastFCA}, and applicable to the two-source case.
The FastFCA does not require frame-wise computation of matrix inverses and matrix products,
and is therefore much faster than the conventional algorithm.
These frame-wise matrix operations are eliminated based on joint diagonalization of the spatial covariance matrices of the source signals. 
This is because the joint diagonalization reduces these matrix operations to mere scalar operations of diagonal entries.
The joint diagonalization is realized by solving
a generalized eigenvalue problem of the spatial covariance matrices of the two source signals.
In the two-source case, the exact joint diagonalization is possible,
and consequently the FastFCA is equivalent to the conventional algorithm, 
whereby causing no degradation in source separation performance compared to the FCA. Currently, the number of sources is limited to two in the FastFCA, 
and 
the extension to more than two sources is regarded as future work.

We follow the following conventions throughout the rest of this paper. Signals are represented in the short-time Fourier transform (STFT) domain
with the time and the frequency indices being $n$ and $f$ respectively. $N$ denotes the number of frames, 
$F$ the number of frequency bins up to the Nyquist frequency, 
$\mathcal{N}(\mathbf{m},\mathbf{R})$ the complex Gaussian distribution with mean $\mathbf{m}$ and covariance matrix $\mathbf{R}$,
$\mathcal{E}$ expectation, 
$\delta_{kl}$ the Kronecker delta,
$\mathbf{0}$ the column zero vector of an appropriate dimension,
$\mathbf{I}$ the identity matrix of an appropriate order,
$\text{diag}(\alpha_1,\alpha_2,\dots,\alpha_D)$ the diagonal matrix of order $D$ with $\alpha_k$ being its $(k,k)$ entry ($k=1,2,\dots,D$),
$(\cdot)^\textsf{T}$ transposition,
$(\cdot)^\textsf{H}$ Hermitian transposition, 
$\text{tr}(\cdot)$ the trace, and
$\det(\cdot)$ the determinant.

\section{Full-rank Spatial Covariance Matrix Analysis (FCA)}
\label{sec:FCA}
This section briefly describes the FCA~\cite{Duong2010}.

Let $\mathbf{y}(n,f)\in\mathbb{C}^I$ be the mixtures observed by $I$ microphones with
the $i$th entry corresponding to the $i$th microphone.
Let $\mathbf{x}_j(n,f)\in\mathbb{C}^I$ be the $j$th  source image, where $j\in\{1,2,\dots,J\}$ denotes the source index and $J$ the number of sources.
In this paper,
we focus on the two-source case ($J=2$).
The observed mixtures  are modeled as the sum of the source images as
$
\mathbf{y}(n,f)=\mathbf{x}_1(n,f)+\mathbf{x}_2(n,f).
$
We deal with the problem of estimating $\mathbf{x}_1(n,f)$ and $\mathbf{x}_2(n,f)$ from $\mathbf{y}(n,f)$. 

In the FCA, the source signal $\mathbf{x}_j(n,f)$ is probabilistically modeled as
$
\mathbf{x}_j(n,f)\sim\mathcal{N}(\mathbf{0},\mathbf{R}_j(n,f))$, where $\mathbf{R}_j(n,f)$ denotes the covariance matrix of $\mathbf{x}_j(n,f)$. 
In the FCA, $\mathbf{R}_j(n,f)$ is parametrized as
\begin{align}
\mathbf{R}_j(n,f)=v_j(n,f)\mathbf{S}_j(f).\label{eq:fullrank}
\end{align}
Here,
$\mathbf{S}_j(f)$ is a time-invariant Hermitian positive-definite (and thus full-rank) matrix called a  spatial covariance matrix, which models the acoustic transfer characteristics of the $j$th source signal. $v_j(n,f)$ is 
a time-variant positive scalar, which models the power spectrum of the $j$th source signal.

The model parameters of the FCA, namely $v_j(n,f)$ and $\mathbf{S}_j(f)$, 
are
 estimated based on the maximization of  the following likelihood:
\begin{align}
\displaystyle &\prod_{n=1}^N\prod_{f=1}^Fp(\mathbf{y}(n,f))=\prod_{n=1}^N\prod_{f=1}^F\Biggl(\frac{1}{\pi^I\det (\mathbf{R}_1(n,f)+\mathbf{R}_2(n,f))}\notag\\
&\times\exp (-\mathbf{y}(n,f)^\textsf{H}(\mathbf{R}_1(n,f)+\mathbf{R}_2(n,f))^{-1}\mathbf{y}(n,f))\Biggr). \label{eq:like}
\end{align}
The likelihood (\ref{eq:like}) can be monotonically increased by an expectation-maximization (EM) algorithm~\cite{Dempster1977}.
The expectation step (E step) updates the conditional expectations 
\begin{align}
\bm{\mu}_j(n,f)&\triangleq\mathcal{E}(\mathbf{x}_j(n,f)\mid\mathbf{y}(n,f)),\\
\bm{\Phi}_j(n,f)&\triangleq\mathcal{E}(\mathbf{x}_j(n,f)\mathbf{x}_j(n,f)^\textsf{H}\mid\mathbf{y}(n,f)),
\end{align}
using the current parameter estimates $v_j^{(l)}(n,f)$ and $\mathbf{S}_j^{(l)}(f)$ by
\begin{align}
\bm{\mu}_j^{(l+1)}(n,f)&= v_j^{(l)}(n,f)\mathbf{S}^{(l)}_j(f)\notag\\
&\phantom{=}\times\Biggl(\sum_{k=1}^2v_k^{(l)}(n,f)\mathbf{S}_k^{(l)}(f)\Biggr)^{-1}\mathbf{y}(n,f),\label{eq:MMSE}\\
\bm{\Phi}_j^{(l+1)}(n,f)&= \bm{\mu}^{(l+1)}_j(n,f)\bm{\mu}^{(l+1)}_j(n,f)^\textsf{H}+v_1^{(l)}(n,f)\mathbf{S}_1^{(l)}(f)\notag\\
&\phantom{\leftarrow}\times\Biggl(\sum_{k=1}^2v_k^{(l)}(n,f)\mathbf{S}_k^{(l)}(f)\Biggr)^{-1}(v_2^{(l)}(n,f)\mathbf{S}_2^{(l)}(f)).\label{eq:postquadexp}
\end{align}
Here, the superscript $(\cdot)^{(l)}$ indicates that this variable is 
computed in the $l$th iteration, and $\triangleq$ means definition.
The maximization step (M step) updates the parameter estimates using $\bm{\Phi}_j^{(l+1)}(n,f)$
by
\begin{align}
v_j^{(l+1)}(n,f)&= \frac{1}{I}\text{tr}(\mathbf{S}_j^{(l)}(f)^{-1}\bm{\Phi}_j^{(l+1)}(n,f)),\label{eq:vupdate}\\
\mathbf{S}^{(l+1)}_j(f)&= \frac{1}{N}\sum_{n=1}^N\frac{1}{v_j^{(l+1)}(n,f)}\bm{\Phi}_j^{(l+1)}(n,f).\label{eq:Supdate}
\end{align}

Once the model parameters have been estimated, the source images can be estimated in various ways. For example, the minimum mean square error (MMSE) estimator of $\mathbf{x}_j(n,f)$ is given by (\ref{eq:MMSE}).

A major drawback of the conventional FCA is expensive computation.
Indeed, the above EM algorithm 
computes matrix inverses and matrix products frame-wise
 in (\ref{eq:MMSE}) and (\ref{eq:postquadexp}).

\section{FastFCA}
\label{sec:fastfca}
This section describes the proposed FastFCA based on joint diagonalization of the spatial covariance matrices $\mathbf{S}_1(f)$ and $\mathbf{S}_2(f)$.
The joint diagonalization eliminates the 
 frame-wise computation of matrix inverses and matrix products, because they
reduce to mere scalar operations of diagonal entries for diagonal matrices.
The joint diagonalization is realized based on
the generalized eigenvalue problem of the matrix pair $\bigl(\mathbf{S}_1(f),\mathbf{S}_2(f)\bigr)$.
See Appendix~\ref{sec:appendix} for mathematical foundations of the generalized eigenvalue problem.

Let $\lambda_1^{(l)}(f), \lambda_2^{(l)}(f), \dots, \lambda_I^{(l)}(f)$
be the generalized eigenvalues of $\bigl(\mathbf{S}_1^{(l)}(f),\mathbf{S}_2^{(l)}(f)\bigr)$, and  
$\mathbf{p}_1^{(l)}(f), \mathbf{p}_2^{(l)}(f), \dots, \mathbf{p}_I^{(l)}(f)$ be generalized eigenvectors of $\bigl(\mathbf{S}_1^{(l)}(f),\mathbf{S}_2^{(l)}(f)\bigr)$ that satisfy
\begin{align}
\begin{cases}
\mathbf{S}_1^{(l)}(f)\mathbf{p}_i^{(l)}(f)=\lambda_i^{(l)}(f)\mathbf{S}_2^{(l)}(f)\mathbf{p}_i^{(l)}(f),\\
\mathbf{p}_i^{(l)}(f)^\textsf{H}\mathbf{S}_2^{(l)}(f)\mathbf{p}_k^{(l)}(f)=\delta_{ik}.
\end{cases}\label{eq:GEVD}
\end{align}
See Appendix~\ref{sec:appendix} for the existence of such $\lambda_1^{(l)}(f), \lambda_2^{(l)}(f), \dots, \lambda_I^{(l)}(f)$ and $\mathbf{p}_1^{(l)}(f), \mathbf{p}_2^{(l)}(f), \dots, \mathbf{p}_I^{(l)}(f)$.
(\ref{eq:GEVD}) can be rewritten in the following matrix forms:
\begin{align}
\begin{cases}
\mathbf{S}_1^{(l)}(f)\mathbf{P}^{(l)}(f)=\mathbf{S}_2^{(l)}(f)\mathbf{P}^{(l)}(f)\mathbf{\Lambda}^{(l)}(f),\\
\mathbf{P}^{(l)}(f)^\textsf{H}\mathbf{S}_2^{(l)}(f)\mathbf{P}^{(l)}(f)=\mathbf{I},
\end{cases}\label{eq:GEVDmatform}
\end{align}
where $\mathbf{P}^{(l)}(f)$ and $\mathbf{\Lambda}^{(l)}(f)$ are defined by
\begin{align}
\mathbf{P}^{(l)}(f)&\triangleq\begin{pmatrix}
\mathbf{p}_1^{(l)}(f)&\mathbf{p}_2^{(l)}(f)&\cdots&\mathbf{p}_I^{(l)}(f)
\end{pmatrix},\\
\mathbf{\Lambda}^{(l)}(f)&\triangleq\text{diag}\bigl(
\lambda_1^{(l)}(f),\lambda_2^{(l)}(f),\cdots,\lambda_I^{(l)}(f)\bigr).
\end{align}
From (\ref{eq:GEVDmatform}), we have
\begin{align}
\mathbf{P}^{(l)}(f)^\textsf{H}\mathbf{S}_1^{(l)}(f)\mathbf{P}^{(l)}(f)&=\mathbf{\Lambda}^{(l)}(f).\label{eq:sigcovdiag}
\end{align}
We see that joint diagonalization of $\mathbf{S}_1^{(l)}(f)$ and $\mathbf{S}_2^{(l)}(f)$ is realized by 
 the transformation
 $\mathbf{P}^{(l)}(f)^\textsf{H}(\cdot)\mathbf{P}^{(l)}(f)$, where $\mathbf{P}^{(l)}(f)$ is obtained based on the generalized eigenvalue problem of $\bigl(\mathbf{S}_1^{(l)}(f),\mathbf{S}_2^{(l)}(f)\bigr)$. 

Now define the following variables that have been basis-transformed by $\mathbf{P}^{(l)}(f)$:
\begin{align}
\tilde{\mathbf{y}}^{(l)}(n,f)&\triangleq\mathbf{P}^{(l)}(f)^\textsf{H}\mathbf{y}(n,f),\label{eq:tildey}\\
\tilde{\bm{\mu}}_j^{(l+1)}(n,f)&\triangleq\mathbf{P}^{(l)}(f)^\textsf{H}\bm{\mu}_j^{(l+1)}(n,f),\label{eq:tildemu}\\
\tilde{\bm{\Phi}}_j^{(l+1)}(n,f)&\triangleq\mathbf{P}^{(l)}(f)^\textsf{H}\bm{\Phi}_j^{(l+1)}(n,f)\mathbf{P}^{(l)}(f),\label{eq:tildePhi}\\
\tilde{\mathbf{T}}_j^{(l)}(f)&\triangleq\mathbf{P}^{(l)}(f)^\textsf{H}\mathbf{S}_j^{(l)}(f)\mathbf{P}^{(l)}(f),\label{eq:tildeS}\\
&=\begin{cases}
\bm{\Lambda}^{(l)}(f),&j=1,\\
\mathbf{I},&j=2,
\end{cases}\label{eq:Sdiag}\\
\tilde{\mathbf{S}}_j^{(l+1)}(f)&\triangleq\mathbf{P}^{(l)}(f)^\textsf{H}\mathbf{S}_j^{(l+1)}(f)\mathbf{P}^{(l)}(f).\label{eq:tildeT}
\end{align}
Here, the tilde indicates the basis transformation. Please be careful about the difference between $(\cdot)^{(l)}$ and $(\cdot)^{(l+1)}$.

The update rules (\ref{eq:MMSE})--(\ref{eq:Supdate}) 
are rewritten in terms of these new variables as in the following, where the indices $n$ and $f$ are omitted for brevity.
\begin{align}
&\tilde{\bm{\mu}}_j^{(l+1)}\notag\\
&=v_j^{(l)}(\mathbf{P}^{(l)})^\textsf{H}\mathbf{S}_j^{(l)}\Biggl(\sum_{k=1}^2v_k^{(l)}\mathbf{S}_k^{(l)}\Biggr)^{-1}\mathbf{y}\ \ \ (\because (\ref{eq:MMSE}), (\ref{eq:tildemu}))\\
&=v_j^{(l)}(\mathbf{P}^{(l)})^\textsf{H}\mathbf{S}_j^{(l)}\underbrace{\mathbf{P}^{(l)}(\mathbf{P}^{(l)})^{-1}}_{\displaystyle\mathbf{I}}\Biggl(\sum_{k=1}^2v_k^{(l)}\mathbf{S}_k^{(l)}\Biggr)^{-1}\notag\\
&\phantom{=}\times \underbrace{
((\mathbf{P}^{(l)})^\textsf{H})^{-1}(\mathbf{P}^{(l)})^\textsf{H}}_{\displaystyle\mathbf{I}}\mathbf{y}\\
&=v_j^{(l)}\tilde{\mathbf{T}}_j^{(l)}\Biggl(\sum_{k=1}^2v_k^{(l)}\tilde{\mathbf{T}}_k^{(l)}\Biggr)^{-1}\tilde{\mathbf{y}}^{(l)}\ \ \ (\because (\ref{eq:tildey}),(\ref{eq:tildeS}))\label{eq:updatemutilde}\\
&=\begin{cases}
v^{(l)}_1\bm{\Lambda}^{(l)}(v^{(l)}_1\bm{\Lambda}^{(l)}+v^{(l)}_2\mathbf{I})^{-1}\tilde{\mathbf{y}}^{(l)},&j=1\\
v^{(l)}_2(v^{(l)}_1\bm{\Lambda}^{(l)}+v^{(l)}_2\mathbf{I})^{-1}\tilde{\mathbf{y}}^{(l)},&j=2
\end{cases}\ \ \ (\because (\ref{eq:tildeS}),(\ref{eq:Sdiag})).\label{eq:MMSEtilde}
\end{align}
\begin{align}
&\tilde{\mathbf{\Phi}}^{(l+1)}_j\notag\\
&=(\mathbf{P}^{(l)})^\textsf{H}\bm{\mu}^{(l+1)}_j(\bm{\mu}^{(l+1)}_j)^\textsf{H}\mathbf{P}^{(l)}
+v_1^{(l)}
(\mathbf{P}^{(l)})^\textsf{H}\mathbf{S}_1^{(l)}\notag\\
&\phantom{\leftarrow}\times\Biggl(\sum_{k=1}^2v_k^{(l)}\mathbf{S}_k^{(l)}\Biggr)^{-1}(v_2^{(l)}\mathbf{S}_2^{(l)})\mathbf{P}^{(l)}\ \ \ (\because (\ref{eq:postquadexp}),(\ref{eq:tildePhi}))\\
&=(\mathbf{P}^{(l)})^\textsf{H}\bm{\mu}^{(l+1)}_j(\bm{\mu}^{(l+1)}_j)^\textsf{H}\mathbf{P}^{(l)}+v_1^{(l)}
(\mathbf{P}^{(l)})^\textsf{H}\mathbf{S}_1^{(l)}\underbrace{\mathbf{P}^{(l)}(\mathbf{P}^{(l)})^{-1}}_{\displaystyle\mathbf{I}}\notag\\
&\phantom{\leftarrow}\times\Biggl(\sum_{k=1}^2v_k^{(l)}\mathbf{S}_k^{(l)}\Biggr)^{-1}\underbrace{
((\mathbf{P}^{(l)})^\textsf{H})^{-1}(\mathbf{P}^{(l)})^\textsf{H}}_{\displaystyle\mathbf{I}}(v_2^{(l)}\mathbf{S}_2^{(l)})\mathbf{P}^{(l)}\\
&=\tilde{\bm{\mu}}^{(l+1)}_j(\tilde{\bm{\mu}}^{(l+1)}_j)^\textsf{H}+v^{(l)}_1v^{(l)}_2\bm{\Lambda}^{(l)}(v^{(l)}_1\bm{\Lambda}^{(l)}+v^{(l)}_2\mathbf{I})^{-1}\label{eq:updatePhitildeeff}\\
&\phantom{=}\notag (\because (\ref{eq:tildemu}),(\ref{eq:Sdiag})).
\end{align}
\begin{align}
&v_j^{(l+1)}\notag\\
&=\frac{1}{I}\text{tr}\Biggl((\mathbf{S}_j^{(l)})^{-1}\underbrace{((\mathbf{P}^{(l)})^\textsf{H})^{-1}(\mathbf{P}^{(l)})^\textsf{H}}_{\displaystyle \mathbf{I}}\bm{\Phi}_j^{(l+1)}\underbrace{\mathbf{P}^{(l)}(\mathbf{P}^{(l)})^{-1}}_{\displaystyle\mathbf{I}}\Biggr)\notag\\
&\phantom{=}(\because (\ref{eq:vupdate}))\\
&=\frac{1}{I}\text{tr}((\tilde{\mathbf{T}}^{(l)}_j)^{-1}\tilde{\bm{\Phi}}^{(l+1)}_j)\ \ \ (\because (\ref{eq:tildePhi}),(\ref{eq:tildeS}))\\
&=\displaystyle\begin{cases}
\displaystyle\frac{1}{I}\text{tr}((\bm{\Lambda}^{(l)})^{-1}\tilde{\mathbf{\Phi}}^{(l+1)}_1),&j=1\vspace{1mm}\\
\displaystyle\frac{1}{I}\text{tr}(\tilde{\mathbf{\Phi}}^{(l+1)}_2),&j=2\\
\end{cases}\ \ \ (\because (\ref{eq:tildeS}),(\ref{eq:Sdiag})).\label{eq:updateveff}
\end{align}
\begin{align}
\tilde{\mathbf{S}}^{(l+1)}_j&=\frac{1}{N}\sum_{n=1}^N\frac{1}{v^{(l+1)}_j}\tilde{\mathbf{\Phi}}^{(l+1)}_j\ \ \ (\because (\ref{eq:Supdate}),(\ref{eq:tildePhi}),(\ref{eq:tildeT})).\label{eq:updateStildeeff}
\end{align}

The generalized eigenvectors $\mathbf{P}^{(l+1)}(f)$ and the generalized eigenvalues $\bm{\Lambda}^{(l+1)}(f)$ of $(\mathbf{S}^{(l+1)}_1(f),\mathbf{S}^{(l+1)}_2(f))$ have also to be computed to be used in  the next iteration. Note that $\mathbf{P}^{(l+1)}(f)$ is needed to compute $\tilde{\mathbf{y}}^{(l+1)}(n,f)$.
One way of doing this is to transform $\tilde{\mathbf{S}}_j^{(l+1)}(f)$ back to $\mathbf{S}_j^{(l+1)}(f)$ by
\begin{equation}
\mathbf{S}_j^{(l+1)}(f)=(\mathbf{P}^{(l)}(f)^\textsf{H})^{-1}\tilde{\mathbf{S}}^{(l+1)}_j(f)\mathbf{P}^{(l)}(f)^{-1}\ \ \ (\because (\ref{eq:tildeT})),\label{eq:recoverS}
\end{equation}
and to solve the generalized eigenvalue problem of $(\mathbf{S}^{(l+1)}_1(f),\mathbf{S}^{(l+1)}_2(f))$.

It is possible to compute
 $\mathbf{P}^{(l+1)}(f)$ and $\bm{\Lambda}^{(l+1)}(f)$ more efficiently
without transforming $\tilde{\mathbf{S}}_j^{(l+1)}(f)$ back to $\mathbf{S}_j^{(l+1)}(f)$.
Indeed, $\mathbf{P}^{(l+1)}(f)$ and $\bm{\Lambda}^{(l+1)}(f)$ can be computed as follows:
\begin{align}
\mathbf{P}^{(l+1)}(f)&=\mathbf{P}^{(l)}(f)\mathbf{Q}^{(l+1)}(f),\label{eq:updatePeff}\\
\bm{\Lambda}^{(l+1)}(f)&=\bm{\Sigma}^{(l+1)}(f),\label{eq:updateLambdaeff}
\end{align}
where $\mathbf{Q}^{(l+1)}(f)$ 
and $\mathbf{\Sigma}^{(l+1)}(f)$ are 
the generalized eigenvectors and the generalized eigenvalues of  
$(\tilde{\mathbf{S}}^{(l+1)}_1(f),\tilde{\mathbf{S}}^{(l+1)}_2(f))$:
\begin{align}
\mathbf{Q}^{(l+1)}(f)&\triangleq\begin{pmatrix}
\mathbf{q}_1^{(l+1)}(f)&\mathbf{q}_2^{(l+1)}(f)&\cdots&\mathbf{q}_I^{(l+1)}(f)
\end{pmatrix},\\
\mathbf{\Sigma}^{(l+1)}(f)&\triangleq\text{diag}\bigl(
\sigma_1^{(l+1)}(f),\sigma_2^{(l+1)}(f),\cdots,\sigma_I^{(l+1)}(f)\bigr).
\end{align}
Here, $\sigma_1^{(l+1)}(f),\sigma_2^{(l+1)}(f),\dots,\sigma_I^{(l+1)}(f)$ denote the generalized eigenvalues  of $\bigl(\tilde{\mathbf{S}}_1^{(l+1)}(f),\tilde{\mathbf{S}}_2^{(l+1)}(f)\bigr)$, and 
$\mathbf{q}_1^{(l+1)}(f),\mathbf{q}_2^{(l+1)}(f),\dots,$\\$\mathbf{q}_I^{(l+1)}(f)$ denote
generalized eigenvectors  of $\bigl(\tilde{\mathbf{S}}_1^{(l+1)}(f),\tilde{\mathbf{S}}_2^{(l+1)}(f)\bigr)$ that satisfy
\begin{align}
\begin{cases}
\tilde{\mathbf{S}}_1^{(l+1)}(f)\mathbf{q}_i^{(l+1)}(f)=\sigma_i^{(l+1)}(f)\tilde{\mathbf{S}}_2^{(l+1)}(f)\mathbf{q}_i^{(l+1)}(f),\\
\mathbf{q}_i^{(l+1)}(f)^\textsf{H}\tilde{\mathbf{S}}_2^{(l+1)}(f)\mathbf{q}_k^{(l+1)}(f)=\delta_{ik}.
\end{cases}\label{eq:GEVDeff5}
\end{align}
Note that (\ref{eq:GEVDeff5}) can also be rewritten in matrix form as follows:
\begin{align}
\begin{cases}
\tilde{\mathbf{S}}_1^{(l+1)}(f)\mathbf{Q}^{(l+1)}(f)=\tilde{\mathbf{S}}_2^{(l+1)}(f)\mathbf{Q}^{(l+1)}(f)\mathbf{\Sigma}^{(l+1)}(f),\\
\mathbf{Q}^{(l+1)}(f)^\textsf{H}\tilde{\mathbf{S}}_2^{(l+1)}(f)\mathbf{Q}^{(l+1)}(f)=\mathbf{I}.
\end{cases}\hspace{-2mm}\label{eq:GEVDeff5matform}
\end{align}

To show (\ref{eq:updatePeff}) and (\ref{eq:updateLambdaeff}), it is sufficient to show
\begin{align}
\begin{cases}
\mathbf{S}^{(l+1)}_1(f)(\mathbf{P}^{(l)}(f)\mathbf{Q}^{(l+1)}(f))\\
=\mathbf{S}^{(l+1)}_2(f)(\mathbf{P}^{(l)}(f)\mathbf{Q}^{(l+1)}(f))\mathbf{\Sigma}^{(l+1)}(f),\\
(\mathbf{P}^{(l)}(f)\mathbf{Q}^{(l+1)}(f))^\textsf{H}\mathbf{S}^{(l+1)}_2(f)(\mathbf{P}^{(l)}(f)\mathbf{Q}^{(l+1)}(f))=\mathbf{I}.
\end{cases}\hspace{-4mm}
\end{align}
This can be shown as follows:
\begin{align}
&\mathbf{S}^{(l+1)}_1(f)\mathbf{P}^{(l)}(f)\mathbf{Q}^{(l+1)}(f)\notag\\
&=\underbrace{(\mathbf{P}^{(l)}(f)^\textsf{H})^{-1}\mathbf{P}^{(l)}(f)^\textsf{H}}_{\displaystyle\mathbf{I}}\mathbf{S}^{(l+1)}_1(f)\mathbf{P}^{(l)}(f)\mathbf{Q}^{(l+1)}(f)\\
&=(\mathbf{P}^{(l)}(f)^\textsf{H})^{-1}\tilde{\mathbf{S}}_1^{(l+1)}(f)\mathbf{Q}^{(l+1)}(f)\ \ \ (\because (\ref{eq:tildeT}))\\
&=(\mathbf{P}^{(l)}(f)^\textsf{H})^{-1}\tilde{\mathbf{S}}_2^{(l+1)}(f)\mathbf{Q}^{(l+1)}(f)\mathbf{\Sigma}^{(l+1)}(f)\ \ \ (\because (\ref{eq:GEVDeff5matform}))\\
&=\mathbf{S}^{(l+1)}_2(f)\mathbf{P}^{(l)}(f)\mathbf{Q}^{(l+1)}(f)\mathbf{\Sigma}^{(l+1)}(f)\ \ \ (\because (\ref{eq:tildeT})),\label{eq:GEVDeff2}\\
&\notag (\mathbf{P}^{(l)}(f)\mathbf{Q}^{(l+1)}(f))^\textsf{H}\mathbf{S}^{(l+1)}_2(f)(\mathbf{P}^{(l)}(f)\mathbf{Q}^{(l+1)}(f))\\
&=\mathbf{Q}^{(l+1)}(f)^\textsf{H}\tilde{\mathbf{S}}^{(l+1)}_2(f)\mathbf{Q}^{(l+1)}(f)\ \ \ (\because (\ref{eq:tildeT}))\\
&=\mathbf{I}\ \ \ (\because (\ref{eq:GEVDeff5matform})).\label{eq:orteff2}
\end{align}

As seen in (\ref{eq:MMSEtilde}) and (\ref{eq:updatePhitildeeff}), the proposed FastFCA
does not require
frame-wise matrix inversion or matrix multiplication, owing to the joint diagonalization.
The additional generalized eigenvalue problem and matrix multiplication in (\ref{eq:updatePeff}) are only required once in each frequency bin per iteration instead of at all time-frequency points, and the FastFCA leads to significantly reduced computation overall.

The algorithm is summarized as follows with $L$ being the number of iterations:

\begin{algorithm}
FastFCA.
\begin{algorithmic}[1]
\STATE Set initial values $v_j^{(0)}(n,f)$, $\mathbf{P}^{(0)}(f)$, and $\bm{\Lambda}^{(0)}(f)$.
\FOR{\ $l= 0$ \ to \ $L-1$}
\STATE Compute $\tilde{\mathbf{y}}^{(l)}(n,f)$ by (\ref{eq:tildey}).
\STATE Compute  $\tilde{\bm{\mu}}_j^{(l+1)}(n,f)$ by (\ref{eq:MMSEtilde}).
\STATE Compute $\tilde{\mathbf{\Phi}}_j^{(l+1)}(n,f)$ by (\ref{eq:updatePhitildeeff}).
\STATE Compute $v_j^{(l+1)}(n,f)$ by (\ref{eq:updateveff}).
\STATE Compute $\tilde{\mathbf{S}}_j^{(l+1)}(f)$ by (\ref{eq:updateStildeeff}).
\STATE Compute $\mathbf{Q}^{(l+1)}(f)$ and $\bm{\Lambda}^{(l+1)}(f)$ by solving the generalized eigenvalue problem of $\bigl(\tilde{\mathbf{S}}_1^{(l+1)}(f),\tilde{\mathbf{S}}^{(l+1)}_2(f)\bigr)$.
\STATE Compute $\mathbf{P}^{(l+1)}(f)$ by (\ref{eq:updatePeff}).
\ENDFOR
\STATE Compute $\bm{\mu}^{(L)}_j(n,f)=(\mathbf{P}^{(L-1)}(f)^\textsf{H})^{-1}\tilde{\bm{\mu}}^{(L)}_j(n,f)$, and output it as the estimate of the source image $\mathbf{x}_j(n,f)$.
\end{algorithmic}
\end{algorithm}

%
%
%
%

\section{source separation Experiment}
\label{sec:exp}
We conducted a source separation experiment to compare
the proposed FastFCA with the conventional FCA~\cite{Duong2010} (see Section~\ref{sec:FCA}).
Both algorithms were implemented in MATLAB (R2013a) running on an Intel i7-2600 3.4-GHz octal-core CPU.
Observed mixtures were generated by convolving 8\,s-long
English speech signals with room impulse responses
measured in a room shown in Fig.~\ref{fig:setting}.
The reverberation time $\text{RT}_{60}$ was 130, 200, 250, 300, 370, or 440\,ms.
Ten trials with different speaker combinations were conducted
for each reverberation time.
The initial values were computed 
based on estimating
the spatial covariance matrices using the time-frequency masks by Sawada's method~\cite{Sawada2011TASLP}.
The source images were estimated in the MMSE sense in both algorithms.
Other conditions are shown in Table~\ref{tab:cond}.

Figure~\ref{fig:RTF} shows the real time factor (RTF) of the EM algorithm for both methods with averaging over all trials and all reverberation times.
Figure~\ref{fig:SDR} shows the signal-to-distortion ratio (SDR)~\cite{Vincent2006} averaged over the two sources and all trials.
The input SDR was 0\,dB.
The FastFCA was over 250 times faster than the FCA with virtually the same SDR.

\begin{figure}[tb]\centering
\includegraphics[width=0.9\columnwidth]{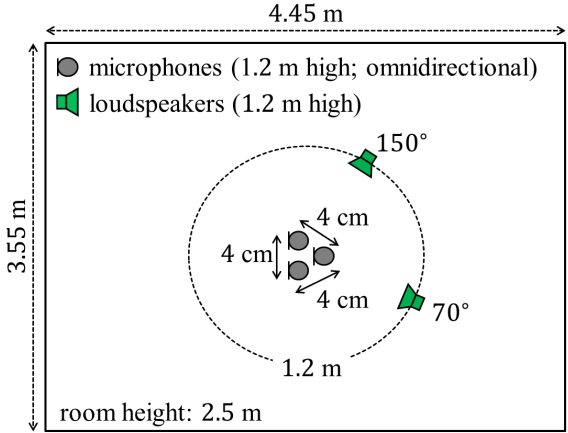}\vspace{-4mm}
\caption{Experimental setting.}\vspace{-3mm}
\label{fig:setting}
\end{figure}

\begin{table}
\centering
\caption{Experimental conditions.}
\label{tab:cond}
\begin{tabular}{ll}\hline
sampling frequency &16\,kHz\\
frame length&1024 (64\,ms)\\
frame shift&512 (32\,ms)\\
window&square root of Hann\\
number of EM iterations&10\\
\hline
\end{tabular}\vspace{-2mm}
\end{table}

\begin{figure}[t]
\centering
\includegraphics[width=\columnwidth]{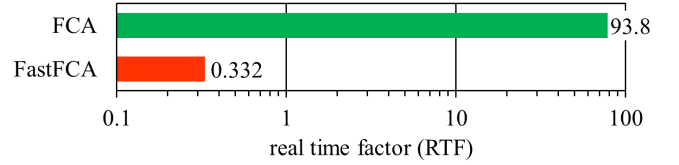}\vspace{-4mm}
\caption{Real time factor (RTF).}
\label{fig:RTF}
\end{figure}

\begin{figure}[tb]
\centering
\includegraphics[width=\columnwidth]{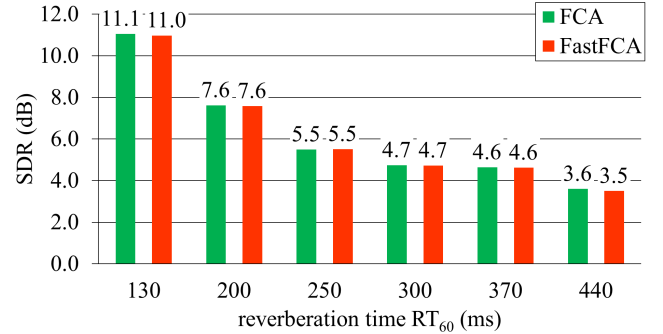}\vspace{-4mm}
\caption{Signal-to-distortion ratio (SDR).}\vspace{-5mm}
\label{fig:SDR}
\end{figure}

\section{Conclusions}
\label{sec:conc}
In this paper, we proposed the FastFCA, a fast algorithm for estimating the FCA parameters in the two-source case with virtually the same SDR as the conventional  algorithm~\cite{Duong2010}. The future work includes 
application to denoising tasks, such as CHiME-3~\cite{Barker2015} and
extension to more than two source signals.

\appendix
\section{Mathematical foundations of the Generalized Eigenvalue Problem}
\label{sec:appendix}
This appendix summarizes mathematical foundations of the generalized eigenvalue problem.
Throughout this appendix, $D$ denotes a positive integer, 
$\mathbf{\Phi}$ and $\mathbf{\Psi}$ complex square matrices of order $D$, and $\lambda$ a complex number.

$\lambda$ is said to be a {\it generalized eigenvalue} of the pair $(\mathbf{\Phi},\mathbf{\Psi})$,
when there exists $\mathbf{p}\in\mathbb{C}^D-\{\mathbf{0}\}$ such that $\mathbf{\Phi p}=\lambda\mathbf{\Psi p}$.
When $\lambda$ is a generalized eigenvalue of $(\mathbf{\Phi},\mathbf{\Psi})$
and $\mathbf{p}\in\mathbb{C}^D-\{\mathbf{0}\}$ satisfies $\mathbf{\Phi p}=\lambda\mathbf{\Psi p}$, $\mathbf{p}$  is said to be a {\it generalized eigenvector} of $(\mathbf{\Phi},\mathbf{\Psi})$ corresponding to $\lambda$.

The polynomial of $\lambda$, $\det (\mathbf{\Phi}-\lambda\mathbf{\Psi})$, is called the {\it characteristic polynomial} of $(\mathbf{\Phi},\mathbf{\Psi})$. It can be shown that
$\lambda$ is a generalized eigenvalue of  $(\mathbf{\Phi},\mathbf{\Psi})$ if and only if $\lambda$ is a root of the 
characteristic polynomial $\det (\mathbf{\Phi}-\lambda\mathbf{\Psi})$.
Indeed, there exists $\mathbf{p}\in\mathbb{C}^D-\{\mathbf{0}\}$ such that $(\mathbf{\Phi}-\lambda\mathbf{\Psi})\mathbf{p}=\mathbf{0}$ if and only if the columns of $\mathbf{\Phi}-\lambda\mathbf{\Psi}$ are linearly dependent, {\it i.e.}, $\det(\mathbf{\Phi}-\lambda\mathbf{\Psi})=0$.

If $\mathbf{\Psi}$ is nonsingular, the fundamental theorem of algebra implies that
the characteristic polynomial $\det (\mathbf{\Phi}-\lambda\mathbf{\Psi})=\det\mathbf{\Psi}\det (\mathbf{\Psi}^{-1}\mathbf{\Phi}-\lambda\mathbf{I})$ has exactly $D$ roots.
In this sense, $(\mathbf{\Phi},\mathbf{\Psi})$ has exactly $D$ generalized eigenvalues.

\begin{theorem}\label{thm}
Suppose $\mathbf{\Phi}$ is Hermitian, 
$\mathbf{\Psi}$ Hermitian positive definite,
and $\lambda_1,\lambda_2,\dots,\lambda_D$ the generalized eigenvalues of $(\mathbf{\Phi},\mathbf{\Psi})$.
There exist $\mathbf{p}_1,\mathbf{p}_2,\dots,\mathbf{p}_D\in\mathbb{C}^D$ such that each $\mathbf{p}_k$ is a generalized eigenvector of 
$(\mathbf{\Phi},\mathbf{\Psi})$ corresponding to $\lambda_k$ and 
$\mathbf{p}_k^\textsf{H}\mathbf{\Psi}\mathbf{p}_l=\delta_{kl}$.
\end{theorem}
\begin{proof}
Since $\mathbf{\Psi}$ is Hermitian positive definite,
there exists a unitary matrix $\mathbf{U}$ and a diagonal matrix $\mathbf{\Sigma}$ with all diagonal entries being positive such that $\mathbf{\Psi}=\mathbf{U\Sigma}\mathbf{U}^\textsf{H}$.
Define a Hermitian matrix $\tilde{\mathbf{\Phi}}$ by $\tilde{\mathbf{\Phi}}=\mathbf{\Sigma}^{-\frac{1}{2}}\mathbf{U}^\textsf{H}\mathbf{\Phi}\mathbf{U}\mathbf{\Sigma}^{-\frac{1}{2}}$. Since
$\det(\tilde{\mathbf{\Phi}}-\mu\mathbf{I})
=\det(\mathbf{\Sigma}^{-\frac{1}{2}}\mathbf{U}^\textsf{H}(\mathbf{\Phi}-\mu\mathbf{U\Sigma}\mathbf{U}^\textsf{H})\mathbf{U}\mathbf{\Sigma}^{-\frac{1}{2}})
=\det(\mathbf{\Psi}^{-1})\det(\mathbf{\Phi}-\mu\mathbf{\Psi})$, 
$\lambda_1,\lambda_2,\dots,\lambda_D$ are the eigenvalues of $\tilde{\mathbf{\Phi}}$. Let $\mathbf{q}_1,\mathbf{q}_2,\dots,\mathbf{q}_D\in\mathbb{C}^D$ be vectors
such that
each $\mathbf{q}_k$ is an eigenvector of 
$\tilde{\mathbf{\Phi}}$ corresponding to $\lambda_k$ and 
$\mathbf{q}_k^\textsf{H}\mathbf{q}_l=\delta_{kl}$.
Define $\mathbf{p}_k$ by $\mathbf{p}_k=\mathbf{U}\mathbf{\Sigma}^{-\frac{1}{2}}\mathbf{q}_k$.
It follows that
$\mathbf{\Phi}\mathbf{p}_k=\mathbf{\Phi}\mathbf{U}\mathbf{\Sigma}^{-\frac{1}{2}}\mathbf{q}_k
=(\mathbf{\Sigma}^{-\frac{1}{2}}\mathbf{U}^\textsf{H})^{-1}\mathbf{\Sigma}^{-\frac{1}{2}}\mathbf{U}^\textsf{H}\mathbf{\Phi}\mathbf{U}\mathbf{\Sigma}^{-\frac{1}{2}}\mathbf{q}_k
=\mathbf{U}\mathbf{\Sigma}^{\frac{1}{2}}\tilde{\mathbf{\Phi}}\mathbf{q}_k
=\lambda_k\mathbf{U}\mathbf{\Sigma}^{\frac{1}{2}}\mathbf{q}_k
=\lambda_k\mathbf{\Psi}\mathbf{p}_k
$.
Furthermore, 
$\mathbf{p}_k^\textsf{H}\mathbf{\Psi}\mathbf{p}_l=\mathbf{q}_k^\textsf{H}\mathbf{q}_l=\delta_{kl}$.
\end{proof}

\bibliographystyle{IEEEbib}

\begin{thebibliography}{1}

\bibitem{Cardoso1999}
J.-F. Cardoso,
\newblock ``High-order contrasts for independent component analysis,''
\newblock {\em Neural Computation}, pp. 157--192, 1999.

\bibitem{Lee1999}
T.-W. Lee, M.~Girolami, and T.J. Sejnowski,
\newblock ``Independent component analysis using an extended infomax algorithm
  for mixed subgaussian and supergaussian sources,''
\newblock {\em Neural Computation}, pp. 417--441, 1999.

\bibitem{Hyvarinen2001}
A.~Hyv{\"a}rinen, J.~Karhunen, and E.~Oja,
\newblock {\em Independent Component Analysis},
\newblock John Wiley {\&} Sons, New York, 2001.

\bibitem{Sawada2005SE}
H.~Sawada, R.~Mukai, S.~Araki, and S.~Makino,
\newblock ``Frequency-domain blind source separation,''
\newblock in {\em Speech Enhancement}, J.~Benesty, S.~Makino, and J.~Chen,
  Eds., pp. 299--327. Springer, Berlin, Heidelberg, 2005.

\bibitem{Duong2010}
N.~Q.~K. Duong, E.~Vincent, and R.~Gribonval,
\newblock ``Under-determined reverberant audio source separation using a
  full-rank spatial covariance model,''
\newblock {\em IEEE Trans. ASLP}, vol. 18, no. 7, pp. 1830--1840, Sep. 2010.

\bibitem{Dempster1977}
A.P. Dempster, N.M. Laird, and D.B. Rubin,
\newblock ``Maximum likelihood from incomplete data via the {EM} algorithm,''
\newblock {\em Journal of the Royal Statistical Society: Series B
  (Methodological)}, vol. 39, no. 1, pp. 1--38, 1977.

\bibitem{Sawada2011TASLP}
H.~Sawada, S.~Araki, and S.~Makino,
\newblock ``Underdetermined convolutive blind source separation via frequency
  bin-wise clustering and permutation alignment,''
\newblock {\em IEEE Trans. ASLP}, vol. 19, no. 3, pp. 516--527, Mar. 2011.

\bibitem{Vincent2006}
E.~Vincent, R.~Gribonval, and C.~F\'{e}votte,
\newblock ``Performance measurement in blind audio source separation,''
\newblock {\em IEEE Trans. ASLP}, vol. 14, no. 4, pp. 1462--1469, Jul. 2006.

\bibitem{Barker2015}
J.~Barker, R.~Marxer, E.~Vincent, and S.~Watanabe,
\newblock ``The third `{CHiME}' speech separation and recognition challenge:
  Dataset, task and baselines,''
\newblock in {\em Proc. ASRU}, Dec. 2015, pp. 504--511.

\end{thebibliography}

\end{document}